\documentclass[runningheads,10pt]{llncs}

\usepackage[english]{babel}
\usepackage{tikz}
\usepackage{latexsym}
\usepackage{amssymb,amsmath}
\usepackage[backref=page]{hyperref}
\usepackage{color}
\hypersetup{
  colorlinks=true,
  citecolor=green,
  linkcolor=red,
  urlcolor=violet}
\usepackage{graphicx}
\usepackage{gastex}
\usepackage[hang,centerlast]{subfigure}

\makeatletter
\def\moverlay{\mathpalette\mov@rlay}
\def\mov@rlay#1#2{\leavevmode\vtop{%
   \baselineskip\z@skip \lineskiplimit-\maxdimen
   \ialign{\hfil$\m@th#1##$\hfil\cr#2\crcr}}}
\newcommand{\charfusion}[3][\mathord]{
    #1{\ifx#1\mathop\vphantom{#2}\fi
        \mathpalette\mov@rlay{#2\cr#3}
      }
    \ifx#1\mathop\expandafter\displaylimits\fi}
\makeatother

\newcommand{\cupdot}{\charfusion[\mathbin]{\cup}{\cdot}}

\title{Opacity with Orwellian Observers
and Intransitive Non-interference\thanks{This research is supported  by the NSERC Discovery Individual Grant No. 13321 (Government of Canada) and   the FQRNT Team Grant No. 167440 (Quebec's Government)  of the first author.}}

\titlerunning{Opacity with Orwellian Observers
and Intransitive Non-interference}

\author{John Mullins\thanks{Corresponding author. Email: {\tt john.mullins@polymtl.ca}} \and Moez Yeddes\thanks{On leave from the department of Mathematics and Computer Science, INSAT, University of Carthage (Tunisia) and  funded by  the CRIAQ-NSERC RDC  Project VerITTAS (AVIO613) No. 435325-12.}}

\authorrunning{John Mullins and Moez Yeddes}

\institute{Dept. of Comp. \& Soft. Eng.Dept.,
\'Ecole Polytechnique de Montr\'eal,
Montreal (Quebec), Canada}

\begin{document}
\maketitle

\begin{abstract}

Opacity is a general behavioural  security scheme flexible enough to account for several specific properties. 
Some secret set of   behaviors  of a system is opaque if a passive attacker can never tell whether the observed behavior  is a secret one or not.  
Instead of considering the case of static observability where the set of observable events is fixed off-line or dynamic observability where the set of observable events changes over time depending on the history of the trace, we consider {\em Orwellian} partial observability where unobservable events are not revealed  unless a downgrading event  occurs in the future of the trace. 
We show how to verify that some regular secret is opaque for a regular language $L$ w.r.t. an {\em Orwellian} projection while it has been proved undecidable even for a regular language $L$ w.r.t. a general  {\em Orwellian observation function}.
We finally  illustrate relevancy of our results by proving the equivalence between the opacity property of regular secrets w.r.t. Orwellian projection and the intransitive non-interference property. 



\end{abstract}



\section{Introduction}
\noindent\emph{Motivations}\,
Opacity has been introduced in~\cite{Mazare2004a} in the context of security protocols and adapted to transition systems in~\cite{Mazare2008}. 
Opacity is a very general security property scheme parametrized by a predicate on system executions (the secret) and an equivalence relation on system executions characterizing the intruder's observation capabilities (each equivalence class corresponds to an observable from the environment). 
The secret is opaque with respect to an observation relation if
any observation class containing a run in the secret, contains also a run that is not in the secret in such a way that an intruder observing the behavior of a system according to the observation relation cannot guarantee to infer from this observation whether the observed trace belongs to the secret or not.
 By adjusting its parameters, opacity can be instantiated to a large class of security or information flow properties (confidentiality, anonymity, non-interference)~(\cite{Mazare2008,Lin2011}).  Observation can be classified as {\em static}, {\em dynamic} and {\em Orwellian} depending on  the computational power  of an observer that it reflects. An observation relation that is {\em static} is defined  {\em a priori} reflecting an observer who always interprets the same event in the same way.  {\em Dynamic observation}  is prefix-based and corresponds to observers able  to deduce knowledge from previous events to interpret the current one that is to each prefix correspond and interpretation of the current action. {\em Orwellian observation} is trace-based and depends at any time not only upon the trace prefix but also upon the trace suffix reflecting observers able to use subsequent knowledge to re-interpret events.  Orwellian observation is required to instantiate opacity to security policies dealing, for instance, with mechanisms for  declassifying or downgrading information. The problem of opacity w.r.t. Orwellian observation is known to be undecidable even for finite transition systems and a degenerated form of opacity~\cite{Mazare2008}. Our aim  is to define a  class of Orwellian observation relations expressive enough to express most of declassification policies like, for instance, intransitive non-interference, when instantiating opacity  to observation relations  in this class, while rendering decidable  the problem of opacity of regular secrets w.r.t. this class  for finite transition systems.

\vspace{0.5em}

\noindent\emph{Contributions}\,
Our contributions are related to the study of  opacity under a class of Orwellian observation relations that  we called {\em Orwellian projections} and its relation to intransitive non-interference.
An Orwellian projection $\pi_{o,d}$ is a natural projection of a sequence of events $\Sigma^\ast$ on observable events $\Sigma_o^\ast$ unless a downgrading event in $\Sigma_d$ occurs subsequently. In this case, the prefix up to this downgrading event is left invariant by  $\pi_{o,d}$.  Our contributions are twofold: first, they provide solutions to the verification problem for opacity w.r.t. $\pi_{o,d}$ and second, they relate  opacity w.r.t. $\pi_{o,d}$ to  another concept used in the formal security community namely, intransitive non-interference (INI).

Concerning the verification problem, our first contribution is a language-theoretic characterization of opacity w.r.t.  $\pi_{o,d}$ in terms of opacity w.r.t.  natural projections (Theorem~\ref{check}).  It has to be noted that this characterization is not in itself an effective procedure since opacity verification w.r.t.
 $\pi_{o,d}$ is reduced to an infinite number of verifications of opacity w.r.t. natural projections. The first step toward an algorithm, our main contribution  to the verification problem (Theorem~\ref{effectivecheck}), is based on Theorem~\ref{check} and the construction of a transition system  incorporating the regular  secret together with the finite transition system modeling the original  system. Opacity verification w.r.t. $\pi_{o,d}$ is then reduced to the verification  of opacity of, in the worst case,  $N$ regular secrets w.r.t. natural projections for  $N$ finite transition systems where $N$ is the number of downgrading transitions of the original system. 
 

Lastly, as an application that illustrates the relevancy of the notion of opacity w.r.t. Orwellian projections, we show that opacity of regular secrets w.r.t. Orwellian projection and intransitive non-interference are reducible to each other (Theorems~\ref{Op2INI} and~\ref{INI2Op}). In order to prove this, we  prove, as a building block, such a characterization of opacity of regular languages w.r.t. natural projections and non-interference (Theorem~\ref{Op2NI}), generalizing, as a side effect, a similar result obtained in~\cite{Mazare2008} for a degenerated form of opacity and non-interference.

\vspace{0.5em}

\noindent\emph{Related work}\, 
Algorithms for verifying opacity in Discrete Event Systems w.r.t. projections are presented together with applications in \cite{Darondeau2007,Takai2009,Hadji2011,Lin2011}.
In~\cite{Darondeau2007}, the authors consider a concurrent version of opacity and show that it is decidable for regular systems and secrets.
In~\cite{Takai2009},  the authors define what they called {\em secrecy } and provide  algorithms for verifying this property. A system property satisfies {\em  secrecy} if the property and its negation are state-based opaque. In~\cite{Lin2011} the author provides an algorithm for verifying state-based opacity (called {\em strong opacity})  and shows how opacity can be instantiated  to important security properties in computer systems and communication protocols, namely anonymity and secrecy. 
In ~\cite{Hadji2011}, the authors define the notion of K-step opacity where the system remains state-based opaque in any step up to depth-k observations that is, any observation disclosing the secret has a length greater than k. Two methods are proposed for verifying  K-step opacity.
In~\cite{Cassez2012}, the authors introduce dynamic projections where the set of events the user can observe changes over time and show how to check that a system is opaque w.r.t. that class of dynamic observation functions. 

{\em Non-interference} (NI) and {\em intransitive non-interference} (INI) for deterministic Mealy machines have been defined in~\cite{rushby:channel.control}. In~\cite{Pinsky1995}, an algorithm is provided for INI. A formulation of INI within the context of non-deterministic LTSs is given in~\cite{mullins00ai}, in the form of a property called {\em admissible interference} (AI), which is verified by reduction to the verification of  a stronger version than NI called {\em strong non-deterministic non-interference} (SNNI) in~\cite{focardi01classification} of $N$ finite transition systems where $N$ is the number of downgrading transitions of the original system. It coincides with INI in the case of deterministic LTS. In~\cite{Bossi04modellingdowngrading}, various notions of INI properties are considered and compared but no comparison with Rushby's original definition is provided. In \cite{Nejib2005a}, the observability theory of discrete event systems is used to formulate and provide an algorithmic approach to the INI verification problem. In \cite{Meyden07}, the author has argued that Rushby's definition of security for intransitive policies that corresponds roughly to  our notion of Orwellian projection suffers from some  flaws, and proposed some stronger variations in the context of deterministic Mealy machines. In \cite{Gorrieri2011}, the authors reformulate Rushby's definition in the setting of deterministic LTSs.  Verification  is then reduced to the verification  of an equivalent characterization  to SNNI called {\em strong non-deducibility on compositions}  (SNDC)  in~\cite{focardi01classification}  of  $N$ finite transition systems where $N$ is the number of high-level transitions in the original system. 

\vspace{0.5em}

\noindent\emph{Paper organization}\,
The paper is organized as follows. The next section presents some preliminaries in labeled transition systems and opacity.
In Section~\ref{Orwelian}, we define and provide a verification algorithm of with respect to Orwellian projections. 
Finally, in section~\ref{INI} we show  the equivalence between the opacity problem w.r.t. Orwellian projections and intransitive non-interference.


\section{Opacity in Transition Systems}
\label{S-Opacity}

\subsection{Labeled transition systems and their languages}


A finite deterministic {\em labeled transition system} over $\Sigma$ is a 4-tuple $G= (\Sigma,Q,\delta,q_0)$,
where $\Sigma$ denotes the alphabet of events, $Q$ denotes the state space, $\delta$ is a partial function from $Q \times \Sigma$ to $Q$, called {\em labeled transition function} and $q_0$ the {\em initial state}. 
The partial function $\delta$ is naturally  extended to a partial function $\delta : Q \times \Sigma^{*} \rightarrow Q$ defined recursively on strings in $\Sigma^{*}$ as follows: $\delta(q,\epsilon)=q$ and $\delta(q,s\cdot \alpha)=\delta(\delta(q,s),\alpha)$ where $\epsilon$ denotes the empty word and, given $s, s' \in \Sigma^\ast$, $s \cdot s'$, the concatenation of $s$ with $s'$. A state $q \in Q$ is {\em reachable} from $q'$ (or simply reachable, if $q'=q_0$)  if there is a word $s \in \Sigma^\ast$ such that $\delta(q', s) = q$. $G$ is {\em finite} if $Q$ and $\Sigma$ are. $G$ is {\em reduced} if any state in $Q$ is reachable. $G$ is complete if $\delta$ is a (total) function.
For $q \in Q$, the transition system obtained from $G$ by starting from $q$ is defined as $G^q = (\Sigma,Q,\delta,q)$. For $\Sigma' \subseteq \Sigma$, the {\em restriction} of $G$ to $\Sigma \setminus \Sigma'$ is defined as $G \setminus \Sigma' = (\Sigma \setminus \Sigma',Q,\delta',q_0)$ where $\delta' = \delta | (\Sigma \setminus \Sigma')$ that is, the restriction of $\delta$ to $\Sigma \setminus \Sigma'$.  Lastly, given two transition systems  $G= (\Sigma,Q,\delta,q_0)$ and  $G'= (\Sigma,Q',\delta',q_0^\prime)$ over the same alphabet $\Sigma$, their {\em product} is the labeled transition system $G \times G' = (\Sigma, Q \times Q', \delta \times \delta', (q_0, q_0^\prime))$ where
$( \delta \times \delta')((q,q'), \alpha) = (\delta(q, \alpha), \delta'(q', \alpha))$.

The concatenation product is extended to languages as follows: given $L, L' \subseteq \Sigma^\ast$, $L \cdot L' = \{s \cdot s' : s\in L \mbox{ and } s' \in L'\}$. $\{s\} \cdot L$ will be simply denoted by $s \cdot L$. In the sequel we let  $s^{-1}L$ denote the left quotient of $L$ by $s$ that is, the set $\{t \in \Sigma^\ast: s\cdot t \in L\}$ and $\overline{ L}$ denote the prefix-closure of $L$ that is, the set $\{ s \in \Sigma^\ast: \exists {t \in \Sigma^\ast} \mbox{ s.t. } s \cdot t \in L\}$. The {\em language of} $G = (\Sigma,Q,\delta,q_0)$ is the set of words
$$L(G)=\{s\in\Sigma^{*}: \delta(q_0,s) \mbox{ is defined}\}.$$
For $F \subseteq Q$, the set of words recognized by states in $F$ is defined as
$L_F(G)=\{s\in\Sigma^{*}: \delta(q_0,s)\in F\}$.


\subsection{Opacity}

We consider a language $L \subseteq \Sigma^\ast$ modelling the behavior of a system and $\Sigma_{o} \subseteq \Sigma$. Opacity qualifies a predicate $\varphi$, given as a subset of $L$, with respect to an \emph{observation function} ${\cal
  O}$ from $\Sigma^\ast$ onto  set $\Sigma_{o}^\ast$ of
\emph{observables} modelling user capabilities for observing the system.  Two strings $s$ and $s'$ of $L$  are equivalent
w.r.t. ${\cal O}$ if they produce the same observable: ${\cal O}(s)
= {\cal O}(s')$.  The set ${\cal O }^{-1} (o) \cap L$ (written $[s]_{\cal O}^L$) is called an
\emph{observation class}.

\begin{definition}[Opacity]
\label{Opacity} 
Given  a language $L \subseteq \Sigma^\ast$, a language $\varphi\subseteq L$ is opaque for $L$ w.r.t. ${\cal O}$ if,
\begin{equation}
\label{cond1} \forall {s \in \varphi}, \exists  {s' \in [s]_{\cal O}^{L}}  \mbox{ s.t. } s'\notin \varphi
\end{equation}
\end{definition}

The information flow deduced by the attacker when the system is not opaque is captured by the notion of secret disclosure. 

\begin{definition} [Secret disclosure]
\label{Disclosure} 
A string $s\in \varphi$ discloses the secret $\varphi$ in $L$ w.r.t. ${\cal O}$  if $[s]_{\cal O}^{L} \subseteq \varphi$.
\end{definition}

\begin{remark}
It follows immediately from the definitions that 
a secret $\varphi\subseteq L$ opaque for $L$ w.r.t. ${\cal O}$ if and only if there is no $s \in \varphi$ disclosing the secret $\varphi$ in $L$ w.r.t. ${\cal O}$.
Equivalently, a secret $\varphi\subseteq L$ opaque for $L$ w.r.t. ${\cal O}$ if and only if ${\cal O}(\varphi) \subseteq {\cal O}(L \setminus \varphi)$.
\end{remark}

Natural projections  provide a class of static observation functions. Given a language $L \subseteq \Sigma^\ast$ and $\Sigma_o  \subseteq \Sigma$, the {\em natural projection} of $L$ on $\Sigma_o $ is the language $\pi_o(L)$ where $\pi_o$ is the function from $\Sigma^\ast$ to $\Sigma_o^\ast$ defined by $\pi_o(\epsilon) = \epsilon$ and
\begin{equation}
\label{P}
\pi_o(s \cdot \alpha) = \left\{ \begin{array}{l@{ }l}
\pi_o(s)\cdot \alpha \ & 
 \ \mbox{if} \ \alpha \in \Sigma_o, \\
  \pi_o(s) &
  \ \mbox{otherwise.}
\end{array}\right.
\end{equation}

In the sequel we will denote $[s]_{\pi_o}^L$ by  $[s]_{o}^L$ for $s \in \Sigma^\ast$  and $L \subseteq \Sigma^\ast$ in order to simplify the notation.



\begin{proposition}[\cite{Darondeau2007}]
Given $L$ and $\varphi$, regular, it is decidable whether $\varphi$ is opaque for $L$ w.r.t. $\pi_o$.
\end{proposition}

\begin{proof}
$\varphi$ is opaque for $L$ w.r.t. $\pi_o$ if and only if $\pi_o(\varphi) \subseteq \pi_o(L \setminus \varphi)$. As $\varphi$ and $L$ are regular, $L \setminus \varphi$ is regular and since $\pi_o$ is a morphism and images under morphisms of regular languages are regular, this relation can be decided.
\end{proof}

\begin{example}

Let $\Sigma = \{h_1, h_2, a, b, c\}$, $\Sigma_o =  \{a, b, c\}$ and $L$, the prefix-closed language accepted by the finite automaton of Fig.~\ref{Fopacity} (where all states are accepting states). Consider the natural projection on $\Sigma_o$  as  observation function. Define the secret language  $\varphi=a^*(b^* +  c^*)$. This secret need to be not deduced by the user of the system, knowing that $h_1$ and $h_2$ are not observable. $\varphi$ is not opaque for $L$ w.r.t. $\pi_{o}$, as by observing $abb$, it is the only one in $[abb]_{{o}}^{L}$. Note that if the attacker observes only $ab$, he can not deduce whether the current sequence of actions of the system belongs to the secret since $[ab]_{{o}}^{L}=\{ab,h_2ab\}$ and $h_2ab \notin \varphi$.

\begin{figure}[!ht]
  \begin{center}
    \unitlength=4pt
    \begin{picture}(25, 26)(0,-12)
    \gasset{Nw=5,Nh=3,Nmr=0,curvedepth=0}
    \thinlines
    \node[linecolor=red!70,fillcolor=red!20](A0)(11,0){$0$}
        \imark(A0)
    \node[linecolor=red!70,fillcolor=red!20](A1)(21,0){$1$}
    \node[linecolor=red!70,fillcolor=red!20](A2)(33,10){$2$}
    \node[linecolor=red!70,fillcolor=red!20](A3)(33,-10){$3$}
    \node[linecolor=red!70,fillcolor=red!20](A4)(11,-10){$4$}
      \node[linecolor=red!70,fillcolor=red!20](A6)(11,10){$6$}
          \node[linecolor=red!70,fillcolor=red!20](A7)(1,10){$7$}
                 \node[linecolor=red!70,fillcolor=red!20](A8)(-9,10){$8$}
    \node[linecolor=red!70,fillcolor=red!20](A5)(1,-10){$5$}
       \node[linecolor=red!70,fillcolor=red!20](A9)(-9,-10){$9$}
    \drawloop[loopangle=90](A1){$a$}
        \drawloop[loopangle=0](A3){$c$}
                \drawloop[loopangle=0](A2){$b$}
                              \drawloop[loopangle=180](A9){$c$}
                                            \drawloop[loopangle=90](A5){$a$}
    \drawedge(A0,A1){$a$}
    \drawedge[ELside=r](A0,A4){$h_1$}
    
    \drawedge(A1,A2){$b$}
    \drawedge(A1,A3){$c$}
      \drawedge(A0,A6){$h_2$}

    \drawedge(A4,A5){$a$}
       \drawedge(A6,A7){$a$}
              \drawedge(A7,A8){$b$}
                  \drawedge(A5,A9){$c$}
    \end{picture}
  \end{center}
  \caption{\em A non opaque system w.r.t. $\pi_{o}$}
\label{Fopacity}
\end{figure}
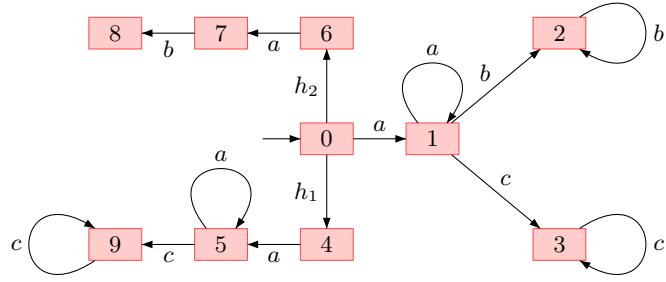

\end{example}

\section{Opacity w.r.t. Orwellian projections}
\label{Orwelian}
 In this section we consider a class of Orwellian observation functions that we call {\em Orwellian projections}. 
Throughout this section, we will consider languages over an alphabet $\Sigma$ partitioned into three sub-alphabets $\Sigma_o, \Sigma_{u}$ and $\Sigma_d$.  $\Sigma_o$ is a set of observable events, $\Sigma_{u}$ a set of unobservable events unless a downgrading event in  $\Sigma_d$ occurs. 

\subsection{Opacity generalized  to Orwellian projections}

We now define the notion of Orwellian projection. 

\begin{definition}
Let $\Sigma_o, \Sigma_d \subseteq \Sigma$. The Orwellian projection on $\Sigma_o$ unless $\Sigma_d$ is the mapping  $ \pi_{o,d} : \Sigma^\ast \rightarrow 
\Sigma^\ast$  defined   by $\pi_{o,d} (\epsilon) = \epsilon$ and
\begin{equation}
\label{O}
\pi_{o,d} (s \alpha) = \left\{ \begin{array}{l@{ }l}
s\alpha \ & 
 \ \mbox{if} \ \alpha \in \Sigma _d, \\
  \pi_{o,d} (s) \alpha  \ &
  \ \mbox{if} \ \alpha \in \Sigma _o, \\
  \pi_{o,d} (s) &
  \ \mbox{otherwise.}
\end{array}\right.
\end{equation}
\end{definition}
In  the sequel we denote $[s]_{\pi_{o,d}}^L$ by  $[s]_{o,d}^L$ for $s \in \Sigma^\ast$  and $L \subseteq \Sigma^\ast$ in order to simplify the notation.

An illustration  of the expressiveness of Orwellian projections to model information declassification or intransitive information flow arising when downgrading information is provided in the following example.

\begin{example}
\label{exemple1}
Consider the automaton $G_2$ given in Figure~\ref{Exemplefigureprincipal} with the set $\Sigma=\Sigma_o\cup\Sigma_{u}\cup\Sigma_d$ where $\Sigma_o=\{l\}$, $\Sigma_{u}=\{h\}$ and $\Sigma_d=\{d\}$. Let  L=$L(G_2)$. Consider now a secret described by the language $\varphi=\{hl\}\cup\{hdhl\}\{l\}^*$. We can check easily that $[hl ]_{{o,d}}^L=\{hl\}$,   $[hdl ]_{{o,d}}^L=\{hdl, hdhl\}$ and $[hdhl^n]_{{o,d}}^L=\{hdhl^n\}$ for any $n > 1$ . 
Hence $\varphi$ is not opaque for $L$ w.r.t.  $\pi_{o,d}$ because $hl$ is the only trace that can be observed as $l$  and $hl$ is in the secret.  Consequently, $hl$ discloses the secret $\varphi$. Moreover, as the first $h$ has been revealed  by downgrading in $hdl$ and any trace in $hdhll^*$, $hdhll^n$, for  $n \in {\mathbb N}$, is the only trace to be observed as $hdll^n$  and $hdhll^n$ is in the secret. Hence $hdhll^n$ discloses the secret also for  $n \in {\mathbb N}$.

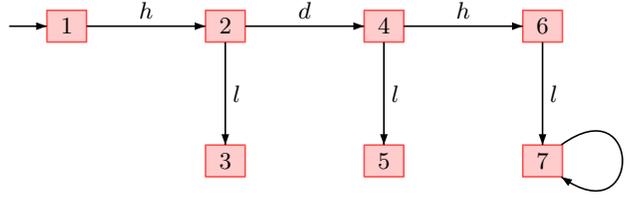
\begin{figure}[!ht]
  \begin{center}
    \unitlength=3pt
    \begin{picture}(60, 25)(-10,-4)
    \gasset{Nw=5,Nh=4,Nmr=0,linewidth=0.20,curvedepth=0}
    \thinlines
    \node[linecolor=red!70,fillcolor=red!20](A1)(-10,17){$1$}
    \imark(A1)
    \node[linecolor=red!70,fillcolor=red!20](A2)(10,17){$2$}
    \node[linecolor=red!70,fillcolor=red!20](A3)(10,0){$3$}
    \drawedge(A1,A2){$h$}
    \drawedge(A2,A3){$l$}
        \node[linecolor=red!70,fillcolor=red!20](A4)(30,17){$4$}
                \node[linecolor=red!70,fillcolor=red!20](A5)(30,0){$5$}
        \node[linecolor=red!70,fillcolor=red!20](A6)(50,17){$6$}
    \node[linecolor=red!70,fillcolor=red!20](A7)(50,0){$7$}
      \drawloop[loopangle=0](A7){$l$}
    \drawedge(A2,A4){$d$}
        \drawedge(A4,A5){$l$}
        \drawedge(A4,A6){$h$}
            \drawedge(A6,A7){$l$}
    \end{picture}
  \end{center}
  \caption{\em A non-opaque system w.r.t. $\pi_{o,d}$}
\label{Exemplefigureprincipal}
\end{figure}

\end{example}

Let us define the smallest set containing the empty string and  prefixes in $\overline{L}$  ending with a downgrading event:
\begin{equation}
\label{Kd}
D(L)=\{\epsilon \} \cup (\overline{L} \cap \Sigma^* \Sigma_d)
\end{equation}
and for each $s\in D(L)$, its continuation  in $(\Sigma\backslash\Sigma_d)^*$:
$$C(s,L)=(\Sigma_o\cup \Sigma_{u})^* \cap s^{-1}L.$$
The following result derives directly: 

\begin{proposition}
\label{decompose}
Let   $\Sigma=\Sigma_o \cup \Sigma_{u} \cup \Sigma_d$ and $L \subseteq \Sigma^\ast$. Then any $u \in L$ admits a {\em unique factorization} $u = st$ such that $s \in  D(L)$ and $t \in  C(s,L)$. In this factorization, $s$ is the longest prefix of L terminating with a $d \in \Sigma_d$.
\end{proposition}

%
%
%

%
%


\begin{remark}
\label{equivalence}
Note that for any $s \in D(L)$,  and $t \in C(s,L)$,  $\pi_{o,d}(st)=s{\pi_{o}}(t)$ as  $s \in \{\epsilon \} \cup  \Sigma^* \Sigma_d$ and  $u  \in (\Sigma_o\cup \Sigma_{u})^*$ and hence, $[st]^L_{o,d} \subseteq s[t]^{C(s,L)}_{o}$ (the reverse inclusion is obvious), linking in that way the Orwellian  projection $\pi_{o,d}$ and the static projection  ${\pi_{o}}$.
\end{remark}

The following theorem reduces opacity disclosure  w.r.t.  $\pi_{o,d}$ to opacity disclosure  w.r.t.  $\pi_{o}$.

\begin{theorem}
\label{check}
A secret $\varphi$ is opaque for $L$ w.r.t. $\pi_{o,d}$ if and only if forall $s \in D(\varphi)$, $C(s,\varphi)$ is opaque for $C(s,L)$ w.r.t. the projection function $\pi_{o}$.
\end{theorem}

\begin{proof} $\Longrightarrow$: \, Suppose that there exists  $s_0 \in D(\varphi)$ such that   $C(s_0,\varphi)$ is not opaque for $C(s_0,L)$ w.r.t.  $\pi_{o}$.
Hence there exists $t_0  \in  C(s_0,\varphi)$ s.t. $t_0$ discloses the secret  $C(s_0,\varphi)$  in $C(s_0,L)$ w.r.t. $\pi_{o}$ and, 
\begin{equation}
\label{from P to O}
\forall {s'\in [t_0]_o^{ C(s_0,L)}}, s'\in  C(s_0,\varphi).
\end{equation}
As $s_0t_0 \in {\varphi}$ 
we have,
\begin{equation}
\label{O=P}
\forall {t'\in[s_0t_0]_{o,d}^L},  \pi_{o,d}(t')=\pi_{o,d}(s_0t_0)=s_0\pi_o(t_0).
\end{equation}
Consequently,
$$ \forall { t'\in[s_0t_0]_{o,d}^L}, t'\in \varphi,$$
 that is, $s_0t_0$ discloses the secret $\varphi$ in $L$ w.r.t. $\pi_{o,d}$, and then the secret $\varphi$ is not opaque for $L$. w.r.t. $\pi_{o,d}$.


\noindent\textbf{$\Longleftarrow:$}\,
We need to prove if $\varphi$ is not opaque for $L$ w.r.t. $\pi_{o,d}$ implies  there exists  $s \in D(L)$ such that  $C(s,\varphi)$ is not opaque for $C(s,L)$ w.r.t.  $\pi_{o}$.

Let $t_0$ disclose the secret  $\varphi$ in $L$ w.r.t. $\pi_{o,d}$, then $t_0 = s_0u_0$ with $s_0 \in  D(\varphi)$ and $u_0 \in C(s,L_ {\varphi}) $. 
To prove the result, we will show that  $u_0$ discloses the secret $C(s_0,\varphi)$ in  $C(s_0,L)$ w.r.t. $\pi_{o}$.
Since $t_0$ discloses the secret  $\varphi$ in $L$ w.r.t. $\pi_{o,d}$, then  for all $t' \in [t_0]^L_{o,d}$, $t'\in \varphi$.  As $ [t_0]^L_{o,d} = s_0[u_0]^{C(s_0,L)}_o$ 
then for all $u' \in [u_0]^{C(s_0,L)}_o$,  $u' \in C(s_0,\varphi)$ and hence, $u_0$ discloses the secret $C(s_0,\varphi)$ in  $C(s_0,L)$ w.r.t. $\pi_{o}$.

%
%
%
%
%
%
%
%
%
%
%
%
%
%
%
\end{proof}


\subsection{Checking opacity of regular secrets w.r.t. Orwellian projections}

The verification procedure is based on a product where the secret is incorporated into the transition system. Let a finite LTS $G = (\Sigma,Q,\delta,q_0)$ and $F \subseteq Q$ s.t. $L_F(G) = L$ and a regular secret $\varphi \subseteq \Sigma^\ast$ that we can consider w.l.o.g. included in $L$ (otherwise one takes $L \cap \varphi$ as secret). First, one constructs a complete deterministic transition system $G_{\varphi} = (\Sigma,Q',\delta',q_0^\prime)$ and a set $F_{\varphi} \subseteq Q'$ s.t. $L_{F_{\varphi}}(G_{\varphi}) = \varphi$. Next, we compute the product $G_\#=G \times G_{\varphi}$,  $F_\# = F \times Q'$  and ${F_{\varphi}}_{\#} = F \times F_{\varphi}$. Then we have $L_{{F_{\varphi}}_\#}(G_\#) =    \varphi$ and, because $G_{\varphi}$ is complete,  $ L_{F_\#}(G_\#) = L$.

\begin{example}
Consider the automaton $G_2= (\Sigma,Q,\delta,q_0)$ of Example~\ref{exemple1} in Figure~\ref{Exemplefigureprincipal} with $L = L(G_2)$ and $\varphi=\{hl\}\cup\{hdhl\}\{l\}^*$. Taking  ${G_2}_{\varphi}$ as the complete deterministic automaton depicted in Figure~\ref{fig:G_Sec}  with $F_{\varphi} = \{3, 7\}$, we get that ${G_2}_{\#} = G_2$,  $F_\# = Q$ and ${F_{\varphi}}_\# = \{3, 7\}$ as depicted in  Figure~\ref{Gdiese}.

\begin{figure}[!ht]
  \begin{center}
    \unitlength=3pt
    \begin{picture}(60, 34)(-10,-9)
    \gasset{Nw=5,Nh=5,Nmr=2.5,linewidth=0.20,curvedepth=0}
    \node[linecolor=blue!70,fillcolor=blue!20](A1)(-10,21){$1$}
    \imark(A1)
    \node[linecolor=blue!70,fillcolor=blue!20](A2)(10,21){$2$}
    \node[Nw=5,Nh=4,Nmr=0,linecolor=red!70,fillcolor=red!20](A3)(10,4){$3$}
    \drawedge(A1,A2){$h$}
    \drawedge(A2,A3){$l$}
        \node[linecolor=blue!70,fillcolor=blue!20](A4)(30,21){$4$}
    \node[linecolor=blue!70,fillcolor=blue!20](A5)(30,4){$5$}
        \node[linecolor=blue!70,fillcolor=blue!20](A6)(50,21){$6$}
    \node[Nw=5,Nh=4,Nmr=0,linecolor=red!70,fillcolor=red!20](A7)(50,4){$7$}
      \drawloop[loopangle=0](A7){$l$}
            \drawloop[loopangle=270](A5){$h, d, l$}
    \drawedge(A2,A4){$d$}
        \drawedge(A4,A6){$h$}
            \drawedge(A6,A7){$l$}
                \drawedge(A4,A5){$d, l$}
                                \drawedge(A2,A5){$h$}
                               \drawedge(A7,A5){$d, h$}
                                                          \drawedge(A6,A5){$d, h$}
                                                              \drawedge(A3,A5){$h, d, l$}
    \gasset{curvedepth=-16.5}
   \drawedge(A1,A5){$d, h$}
    \gasset{curvedepth=5}

    \end{picture}
  \end{center}
  \caption{\em ${G_2}_{\varphi}$ for $\varphi=\{hl\}\cup\{hdhl\}\{l\}^*$ with $F_{\varphi} = \{3, 7\}$}
\label{fig:G_Sec}
\end{figure}
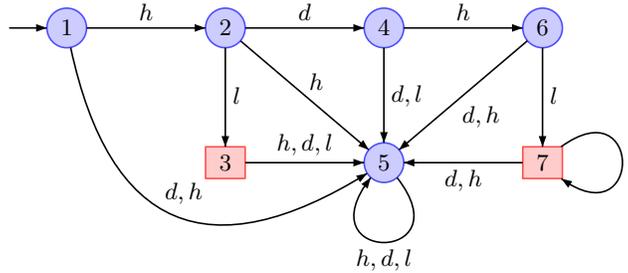

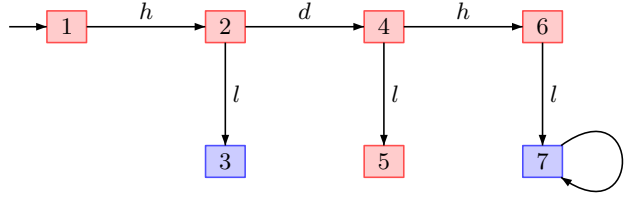
\begin{figure}[!ht]
  \begin{center}
    \unitlength=3pt
    \begin{picture}(60, 24)(-10,-3)
    \gasset{Nw=5,Nh=4,Nmr=0,linewidth=0.20,curvedepth=0}
    \thinlines
    \node[linecolor=red!70,fillcolor=red!20](A1)(-10,17){$1$}
    \imark(A1)
    \node[linecolor=red!70,fillcolor=red!20](A2)(10,17){$2$}
    \node[linecolor=blue!70,fillcolor=blue!20](A3)(10,0){$3$}
    \drawedge(A1,A2){$h$}
    \drawedge(A2,A3){$l$}
        \node[linecolor=red!70,fillcolor=red!20](A4)(30,17){$4$}
                \node[linecolor=red!70,fillcolor=red!20](A5)(30,0){$5$}
        \node[linecolor=red!70,fillcolor=red!20](A6)(50,17){$6$}
    \node[linecolor=blue!70,fillcolor=blue!20](A7)(50,0){$7$}
      \drawloop[loopangle=0](A7){$l$}
    \drawedge(A2,A4){$d$}
        \drawedge(A4,A5){$l$}
        \drawedge(A4,A6){$h$}
            \drawedge(A6,A7){$l$}
    \end{picture}
  \end{center}
    \caption{\em ${G_2}_{\#}$ with $F_\# = Q$ and ${F_{\varphi}}_\# = \{3, 7\}$  for $\varphi=\{hl\}\cup\{hdhl\}\{l\}^*$}
\label{Gdiese}
\end{figure}

\end{example}

 Hence, in the rest of this paper and w.l.o.g., we assume that  $G = G_\#$, $F = F_\#$ and $F_{\varphi} = {F_{\varphi}}_\#$ and thus $L_{F_{\varphi}}(G) = \varphi$ and $L_{F}(G) = L$. Let $Q_d = \{q_0\} \cup \{q \in Q : \exists {q' \in Q} \mbox{ s.t. } \delta(q', d) = q\}$. 
 
 We are now ready for the verification result establishing, as a consequence, the decidability of the opacity verification problem w.r.t. $\pi_{o,d}$ of regular secrets for regular languages.

\begin{theorem}
\label{effectivecheck}
$\varphi$ is opaque for $L$ w.r.t. $\pi_{o,d}$ iff for all $q \in Q_d$, $L_{F_{\varphi}}(G^q \setminus \Sigma_d)$ is opaque for $L_F(G^q \setminus \Sigma_d)$ w.r.t. $\pi_{o}$.
\end{theorem}

\begin{proof}
Suppose that $\varphi$ is not opaque for $L$ w.r.t. $\pi_{o,d}$ then, by Theorem~\ref{check}, for some $s_0 \in D(\varphi)$ and $t_0 \in C(s_0, \varphi)$, $t_0$ discloses the secret  $C(s_0, \varphi)$ in $C(s_0, L)$ w.r.t. $\pi_{o}$. Let $\overline{q} = \delta(q_0, s_0)$. Clearly, $\overline{q} \in Q_d$ and since $L_{F_{\varphi}}(G^{\overline{q}}\setminus \Sigma_d) = C(s_0, \varphi)$ and  $L_{F}(G^{\overline{q}}\setminus \Sigma_d) =  C(s_0, L)$, it turns out that $L_{F_{\varphi}}(G^{\overline{q}} \setminus \Sigma_d)$ is not opaque for $L_F(G^{\overline{q}} \setminus \Sigma_d)$ w.r.t. $\pi_{o}$.  

Conversely, suppose that for some ${\overline{q}} \in Q_d$,  $L_{F_{\varphi}}(G^{\overline{q}} \setminus \Sigma_d)$ is not opaque for $L_F(G^{\overline{q}} \setminus \Sigma_d)$ w.r.t. $\pi_{o}$.  Hence, there are  some $s_0 \in \Sigma^\ast \cdot \Sigma_d \cup \{\epsilon\}$ and  $t_0 \in L_{F_{\varphi}}(G^{\overline{q}} \setminus \Sigma_d)$ s.t. $\delta(q_0, s_0) = \overline{q}$ (since $\overline{q} \in Q_d$) and $t_0$ discloses the secret $L_{F_{\varphi}}(G^{\overline{q}} \setminus \Sigma_d)$ in  $L_F(G^{\overline{q}} \setminus \Sigma_d)$ w.r.t. $\pi_{o}$. Clearly,  $L_{F_{\varphi}}(G^{\overline{q}} \setminus \Sigma_d) = C(s_0, \varphi)$ and   $L_{F}(G^{\overline{q}}\setminus \Sigma_d) =  C(s_0, L)$. Since $s_0 \cdot[t_0]^{C(s_0,L)}_o = [s_0 \cdot t_0]^{L}_{o,d}$ and  $s_0 \cdot t_0 \in \varphi$, $s_0 \cdot t_0$ discloses the secret in $\varphi$ w.r.t. $\pi_{o,d}$ (Theorem~\ref{check}).
\end{proof}
\begin{example}
Consider again the LTS $G_2= (\Sigma,Q,\delta,q_0)$ of Example~\ref{exemple1} in Figure~\ref{Exemplefigureprincipal} with $L = L(G_2)$ and $\varphi=\{hl\}\cup\{hdhl\}\{l\}^*$.  In this case, $Q_d = \{1, 4\}$. $\varphi$ is not opaque for $L$ w.r.t. $\pi_{o,d}$ because $G_2^{1}\setminus{\{d\}}$ is not opaque w.r.t. $\pi_{o}$ since $hl$ discloses the secret. This reflects the case where there is no downgrading along the run (Figure~\ref{fig:fpasmith1}). But this is also the case that  $G_2^{4}\setminus{\{d\}}$ is not opaque w.r.t. $\pi_{o}$ since any sequence in $hll^*$ discloses the secret after downgrading. This reflects that sequences in $hdhll^*$ discloses the secret w.r.t. $\pi_{o,d}$ (Figure~\ref{fig:fpasmith2}).
\begin{figure}
  \centering
\subfigure[Non-opacity of $\{hl\}$ for $G_2^{1}\setminus{\{d\}}$ w.r.t. $\pi_{o}$]{\label{fig:fpasmith1}
    \unitlength=3pt
    \begin{picture}(30,25)(-10,-6)
    \gasset{Nw=5,Nh=4,Nmr=0,linewidth=0.20,curvedepth=0}
    \thinlines
    \node[linecolor=red!70,fillcolor=red!20](A1)(-10,17){$1$}
    \imark(A1)
    \node[linecolor=red!70,fillcolor=red!20](A2)(10,17){$2$}
    \node[linecolor=blue!70,fillcolor=blue!20](A3)(10,0){$3$}
    \drawedge(A1,A2){$h$}
    \drawedge(A2,A3){$l$}
    \end{picture}
}
\subfigure[Non-opacity of $\{hll^\ast\}$ for $G_2^{4}\setminus{\{d\}}$ w.r.t. $\pi_{o}$]{\label{fig:fpasmith2}
    \unitlength=3pt 
    \begin{picture}(30,25)(26,-6)
    \gasset{Nw=5,Nh=4,Nmr=0,linewidth=0.20,curvedepth=0}
    \thinlines
        \node[linecolor=red!70,fillcolor=red!20](A4)(30,17){$4$}
              \imark(A4)
                \node[linecolor=red!70,fillcolor=red!20](A5)(30,0){$5$}
        \node[linecolor=red!70,fillcolor=red!20](A6)(50,17){$6$}
    \node[linecolor=blue!70,fillcolor=blue!20](A7)(50,0){$7$}
      \drawloop[loopangle=0](A7){$l$}
        \drawedge(A4,A5){$l$}
        \drawedge(A4,A6){$h$}
            \drawedge(A6,A7){$l$}
    \end{picture}
}
\caption{Opacity of $\varphi=\{hl\}\cup\{hdhl\}\{l\}^*$ for $G_2$ w.r.t. $\pi_{o,d}$.}
\label{fig:exsmithfpa}
\end{figure}
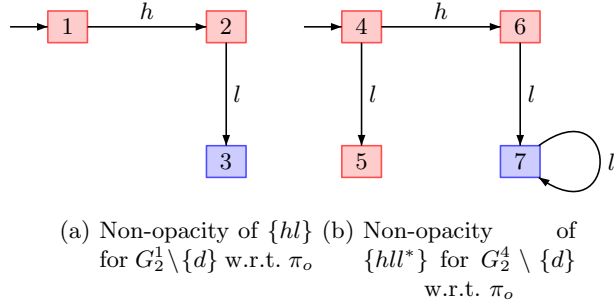
\end{example}

\section{Opacity w.r.t. Orwellian Projection and Intransitive Non-Interference}
\label{INI}

In this section, we  show how our notion of opacity w.r.t. Orwellian projection relates to transitive and intransitive non-interference. For (transitive) non-interference, the alphabet of events  $\Sigma$ is partitioned into  two sets, $High$ (private actions) and $Low$ (public actions). A system is {\em non-interferent} if it is not possible for a public observer to  infer information from the public actions about the presence of private actions in the original run (See~\cite{rushby:channel.control} and~\cite{focardi01classification} for a discussion on non-interference).

\begin{definition}
Let a  LTS  $G= (\Sigma,Q,\delta,q_0)$  and $F  \subseteq Q$ s.t. $L = L_F(G)$ then $L$  satisfies {\em non-interference} (NI) if $\pi_{Low}(L)\subseteq L$.
\end{definition}

For intransitive non-interference, $\Sigma$ is partitioned into three sets $High$ (private actions), $Low$ (public actions) and $Down$ (downgrading actions). A system is intransitive non-interferent if it non-interferent unless a downgrading action occurs and discloses all private actions encountered so far. A discussion on intransitive non-interference can be found in~\cite{rushby:channel.control,Bossi04modellingdowngrading} and~\cite{Meyden07}.
\begin{definition}
Let a  LTS  $G= (\Sigma,Q,\delta,q_0)$  and $F  \subseteq Q$ s.t. $L = L_F(G)$ then $L$  satisfies {\em intransitive non-interference} (INI) if $\pi_{High, Down}(L)\subseteq L$.
\end{definition}

The following result reduces the INI verification problem  to $N$ instances of the NI verification problem where $N$ is the number of downgrading transitions in $G$.
\begin{proposition}
\label{INI2NI}
Let a  LTS  $G= (\Sigma,Q,\delta,q_0)$  and $F  \subseteq Q$ s.t. $L = L_F(G)$.  $L$ satisfies INI iff for all $q \in Q_d$, $L_F(G^q \setminus Down)$ satisfies NI.
\end{proposition}
This result restates in a language-theoretic setting, a result due to~\cite{mullins00ai}.
As a first result in this section, with the aim to build a reduction from opacity of regular secrets w.r.t. an Orwellian projection to INI, we extend  a result due to~\cite{Mazare2008} and reducing  a degenerated form of opacity w.r.t. natural projection to NI.
\begin{theorem}
\label{Op2NI}
The opacity verification  problem of regular secrets  w.r.t. $\pi_o$ for regular languages is reducible to the NI verification problem for finite systems.
\end{theorem}

\begin{proof}
Let $G= (\Sigma,Q,\delta,q_0)$ be a transition system and $F, F_{\varphi} \subseteq Q$  s.t. $L = L_F(G)$ and $\varphi = L_{F_{\varphi}}$.
We construct a labeled transition system  $G^\flat= (\Sigma^\flat,Q^\flat,\delta^\flat,q^\flat_0)$ and $F^\flat \subseteq Q^\flat$ s.t. 
\vspace*{-0.5em}
\begin{eqnarray*}
\Sigma^\flat& = & \Sigma_o \cupdot \{h\} \\
Q^\flat & = & (Q \times \{0\})  \cup (F_{\varphi} \times \{1\}) \\
\delta^\flat &  = &  \{((q, 0), \alpha, (q',0)) : (q, \alpha, q') \in \delta \mbox{ and } \alpha \in \Sigma_o\} \cup \\
&&  \{((q, 0), \epsilon , (q',0)) : (q, \alpha, q') \in \delta \mbox{ and } \alpha \in \Sigma_u\} \cup \\
&& \{((q, 0), h , (q,1)) : q \in F_{\varphi}\}  \\
q^\flat_0 & = & (q_0, 0) \\
F^\flat & = & (F \cap (Q \setminus F_{\varphi}) \times \{0\})\cup   (F_{\varphi} \times \{1\}).
\end{eqnarray*}
Now, we consider a non-interference problem for $G^\flat$ with the following partitioning:
$Low  =  \Sigma_o \cup \{\epsilon\}$ and $High  =  \{h\}$
 and finally, we show  that $\varphi$ is opaque for $L$ w.r.t. $\pi_o$ iff $L^\flat= L_{F^\flat}(G^\flat)$ satisfies NI. 
 
 $\Longrightarrow$: \, Suppose that $L^\flat = L_{F^\flat}(G^\flat)$ does not satisfies NI then there is  $s^\prime_0 \in  \pi_{Low}(L^\flat) = L_{F_{\varphi} \times \{0\}}(G^\flat)$ s.t. $s^\prime_0 \not \in L^\flat$.  Thus  $s^\prime_0\not  \in L_{F \cap (Q \setminus F_{\varphi})}(G^\flat)$. Hence,  for  $s_0 \in L_{F}(G)$ s.t. $\pi_o(s_0) = s^\prime_0$, $s_0 \in L_{F_{\varphi}}(G)$ but for any $s \not \in \varphi$, $\pi_o(s) \not = \pi_o(s_0)$, that is, $s_0$ discloses $\varphi$ w.r.t. $\pi_o$.
 
  $\Longleftarrow$: \, Suppose that $\varphi$ is not opaque for $L$ w.r.t. $\pi_o$ then for some $s_0 \in L$, $s_0$ discloses  the secret $\varphi$
 that is, $[s_0]^L_{\pi_o}  \subseteq \varphi$, $s^\prime_0 = \pi_o(s_0)  \cdot h \in L_{F_{\varphi} \times \{1\}}(G^\flat) \subseteq L^\flat$ and hence  $\pi_{Low}(s^\prime_0)=  \pi_o(s_0)  \in \pi_{Low}(L^\flat)$ and $\pi_{Low}(s^\prime_0)  \not\in L^\flat$ since  $ \pi_o(s_0)  \not \in L_{F \cap (Q \setminus F_{\varphi})}(G^\flat)$. Consequently $\pi_{Low}(L^\flat)\not \subseteq L^\flat$.
\end{proof}

%

Now, we build a  reduction of opacity w.r.t. an Orwellian projection  to INI by using the previous reduction as building blocks:

\begin{theorem}
\label{Op2INI}
Opacity verification of regular secrets w.r.t. $\pi_{o,d}$ for regular languages is reducible to INI verification of finite systems.

\end{theorem}
\begin{proof}
Let $G= (\Sigma,Q,\delta,q_0)$ be a transition system and $F, F_{\varphi} \subseteq Q$  s.t. $L = L_F(G)$ and $\varphi = L_{F_{\varphi}}$.
We construct a labeled transition system  $G^\natural= (\Sigma^\natural,Q^\natural,\delta^\natural,q^\natural_0)$ and $F^\natural \subseteq Q'$ s.t. 
\vspace*{-.5em}
\begin{eqnarray*}
\Sigma^\natural & = & \Sigma_o \cupdot  \Sigma_d \cupdot \{h\} \\
Q^\natural  & = & (Q \times \{0\})  \cup (F_{\varphi} \times \{1\}) \\
\delta^\natural &  = &  \{((q, 0), \alpha, (q',0)) : (q, \alpha, q') \in \delta \mbox{ and } \alpha \in \Sigma_o \cup \Sigma_d\} \cup \\
&&  \{((q, 0), \epsilon , (q',0)) : (q, \alpha, q') \in \delta \mbox{ and } \alpha \in \Sigma_u\} \cup \\
&& \{((q, 0), h , (q,1)) : q \in F_{\varphi}\} \\
q^\natural_0 & = & (q_0, 0) \\
F^\natural & = & (F \cap (Q \setminus F_{\varphi}) \times \{0\})\cup   (F_{\varphi} \times \{1\}).
\end{eqnarray*}
Now, we consider an intransitive non-interference problem for $G^\natural$ with the partitioning  $Low  =  \Sigma_o$, 
  $Down  =  \Sigma_d$ and 
   $High  = \{h\}$, and we show  that $\varphi$ is opaque for $L$ w.r.t. $\pi_{o,d}$ iff $L^\natural= L_{F^\natural}(G^\natural)$ satisfies INI. 
 It has first to be noted that, by construction,  for any $q \in Q_d$, 
 \begin{equation}
 \label{eq1}
 L_{F^\flat}((G^q \setminus \Sigma_d)^\flat) = L_{F^\natural}({G^\natural}^{(q, 0)} \setminus Down)
 \end{equation}
 and second,  that for any $q  \in F_{\varphi}$, $(q, 1) \not \in Q^\natural_d$ since the only transition going into these states is an $h$-transition. Hence,
 \begin{equation}
 \label{eq2}
Q^\natural_d =  Q_d \times \{0\}.
 \end{equation}
 
 Also, $\varphi  \mbox{ is  opaque for } L  \mbox{ w.r.t. } \pi_{o,d}$ iff   for any ${q \in Q_d}$,
  \begin{align*}
L_{F_{\varphi}}(G^q \setminus \Sigma_d) \mbox{ is opaque for } L_F(G^q \setminus \Sigma_d) \mbox{ w.r.t. } \pi_o \\
 \mbox{(by Theorem~\ref{effectivecheck})}  & \Longleftrightarrow\\
 L_{F^\flat}((G^q \setminus \Sigma_d)^\flat) \mbox{ satisfies NI }  \mbox{(by Theorem~\ref{Op2NI})}   &  \Longleftrightarrow \\
 L_{F^\natural}({G^\natural}^{(q, 0)} \setminus Down) \mbox{ satisfies NI } \mbox{(by Eq.~\ref{eq1})} &  \Longleftrightarrow
 \end{align*}
 for all ${q \in Q^\natural_d}$,
   \begin{align*}
   L_{F^\natural}({G^\natural}^{q} \setminus Down) \mbox{ satisfies NI } & \mbox{(by Eq.~\ref{eq2})} &   \Longleftrightarrow\\
  L_{F^\natural}({G^\natural})  \mbox{ satisfies INI } &  \mbox{(by Proposition~\ref{INI2NI})}   \\
 \end{align*}
 
%

\end{proof}

Finally, we define a reverse reduction extending a similar reduction defined in~\cite{Mazare2008} from NI to opacity w.r.t. natural projection.
\begin{theorem}
\label{INI2Op}
INI verification for finite systems is reducible to   opacity verification of regular secrets w.r.t. $\pi_{Low, Down}$ for regular languages.
\end{theorem}
\begin{proof}
Let $G= (\Sigma,Q,\delta,q_0)$ be a transition system with $\Sigma = High \cupdot Low \cupdot Down$, $F, \subseteq Q$  s.t. $L = L_F(G)$ and $\varphi = \{s \in L_F(G) : \pi_{Low, Down}(s) \not = s\}$. We show that $L$ satisfies INI iff $\varphi$ is opaque for $L$ w.r.t. $\pi_{Low, Down}$.

  $\Longrightarrow$:  \,  Suppose that $L= L_{F}(G)$ does not satisfies INI, then for some $u_0 \in L$, $\pi_{o,d}(u_0) \not \in L$. Hence, $u_0 \in \varphi$ (since $\pi_{o,d}(u_0) \not= u_0$) and $[u_0]_{o,d}^L \subseteq \varphi$ (otherwise, $\pi_{o,d}(u_0) \in L$, getting a contradiction). Consequently, $u_0$ discloses $\varphi$.


  $\Longleftarrow$: \, Suppose that $s \in \varphi$, then $\pi_{o,d}(s) \not = s$ and also,  there is an $s' \in L$ s.t. $\pi_{o,d}(s') = \pi_{o,d}(s)$ and $\pi_{o,d}(s') \in L$ since $L$ satisfies INI. Moreover  $\pi_{o,d}(s')  \not \in \varphi$ since $\pi_{o,d}(\pi_{o,d}(s')) = \pi_{o,d}(s')$. Thus $\varphi$ is opaque w.r.t. $\pi_{o,d}$ for $L$. 

%


%
%
\end{proof}

\section{Conclusion}
In this paper, we have investigated the opacity verification problem in the context of finite systems, regular secrets and a class of Orwellian observation functions that we called Orwellian projections. As an illustration of the relevancy of this problem in the context of the verification of information flow properties in the domain of security-critical systems, we have related opacity w.r.t. Orwellian projections to INI for finite systems by showing a computational equivalence between both notions, providing , as a side effect a characterization of NI for finite systems with opacity of regular secrets w.r.t. natural projections  for regular languages. 

We are now investigating the opacity synthesis problem consisting of compute the supremal opaque sublanguage w.r.t. Orwellian projection of a given language, and its dual problem, which is also very challenging since, in this case, there is a large range of modifications to the initial system that can be considered, e.g., enlarging the behavior of the non-secret part inserting suitable downgrading actions whenever possible or cutting some possible secret behaviors. In future works, we will investigate the problem of supervisory control for opacity w.r.t. Orwellian projections along a line of research initiated by~\cite{Dub2010} for opacity w.r.t. natural projections. We also plan to instantiate opacity to {\em Intransitive Non-interference with Selective Declassification} (INISD) which has been suggested recently in~\cite{Gorrieri2011}. INISD generalises INI  by allowing to each downgrading action $d$ to declassify only a subset $H(d)$ of non-observable events, which is more likely to be of practical interest. A structural definition of this property for Petri nets  has been proposed and its decidability has been investigated in~\cite{eike-darondeau2012}.

\bibliographystyle{splncs}


\end{document}